\documentclass[a4paper,reqno, 11pt]{amsart}

\usepackage{fullpage}

\usepackage[utf8]{inputenc}

\usepackage{lmodern}
\usepackage{array}
\usepackage{graphicx}
\usepackage{stmaryrd}
\usepackage{subfig}
\usepackage{amssymb, amsfonts, amsthm, amsmath}
\usepackage{enumerate}
\usepackage{enumitem}
\usepackage{bbm}
\usepackage{enumitem}
\usepackage{hyperref}
\usepackage{cleveref}
\usepackage[english]{babel}

\usepackage{tikz}
\usetikzlibrary{decorations, calc}

\usetikzlibrary{decorations.pathreplacing}

\numberwithin{equation}{section}

\theoremstyle{plain}
\newtheorem{theorem}{Theorem}[section]

\newtheorem{proposition}[theorem]{Proposition}

\theoremstyle{definition}
\newtheorem{definition}[theorem]{Definition}

\theoremstyle{remark}
\newtheorem{remark}[theorem]{Remark}

\newcommand{\Z}{\mathbb{Z}}
\newcommand{\R}{\mathbb{R}}
\newcommand{\C}{\mathbb{C}}

\usepackage{xcolor}

\renewcommand{\epsilon}{\varepsilon}

\newcommand{\Addresses}{{
  \bigskip
  \footnotesize

  \textsc{\address{D\'epartement de math\'ematiques et applications de l’ENS, Ecole normale sup\'erieure PSL Research University, CNRS UMR 8553, Paris. On~leave from St. Petersburg Department of Steklov Mathematical Institute RAS, Russia.  \textit{Current address:} Department of Mathematics, University of Michigan, 530 Church Street, Ann Arbor, MI 48109, USA.}}\par\nopagebreak
  \textit{E-mail address}: \texttt{dchelkak at umich.edu}

 \medskip

  \textsc{Université Paris-Saclay, CNRS, CEA, Institut de physique théorique, 91191 Gif-sur-Yvette, France}\par\nopagebreak
  \textit{E-mail address}: \texttt{sanjay.ramassamy at ipht.fr}

}}
\providecommand{\keywords}[1]{\textit{Keywords:} #1}
\title{Fluctuations in the Aztec diamonds\\ via a space-like maximal surface in Minkowski 3-space}
\author{Dmitry Chelkak}
\author{Sanjay Ramassamy}
\date{\today}

\newcommand\cG{\mathcal{G}}
\newcommand\cT{\mathcal{T}}
\newcommand\cO{\mathcal{O}}

\begin{document}

\begin{abstract} We provide a new description of the scaling limit of dimer fluctuations in homogeneous Aztec diamonds via the intrinsic conformal structure of a space-like maximal surface in the three-dimensional Minkowski space $\R^{2,1}$. This surface naturally appears as the limit of the graphs of origami maps associated to symmetric t-embeddings of Aztec diamonds, fitting the framework recently developed in~\cite{CLR2,CLR1}.
\end{abstract}

\maketitle

\section{Introduction}
\label{sec:intro}

In this paper we discuss the homogeneous bipartite dimer model on Aztec diamonds~\cite{EKLP-I}. Though this is a \emph{free fermion} model, it is still interesting due to the presence of a \emph{non-trivial conformal structure} generated by specific boundary conditions. We refer to a recent paper~\cite{ADSV} and references therein for a discussion of this conformal structure from the theoretical physics perspective; see also~\cite{GBDJ} for a similar discussion in the interacting fermions context. Mathematically, this is one of the most classical and rigorously studied examples demonstrating the rich behavior of the dimer model: convergence of fluctuations to the Gaussian Free Field (GFF) in a non-trivial metric, Airy-type asymptotics near the curve separating the phases, etc. We refer the interested reader to lecture notes~\cite{gorin-lectures,kenyon-lectures} for more information on the subject.

The Aztec diamond~$A_n$ of size~$n$ is a subset~$|p|+|q|\le n$ of the square grid~$(p,q)\in(\mathbb{Z}+\tfrac{1}{2})^2$; see Fig.~\ref{fig:A-reduction}. The dimer model on~$A_n$ is a random choice of a perfect matching in~$A_n$. Let
\[
\diamondsuit:=\{(x,y)\in\R^2:|x|+|y|<1\}\quad \text{and}\quad \mathbb{D}_\diamondsuit:=\tfrac{\sqrt{2}}{2}\mathbb{D}=\{(x,y)\in\R^2:x^2+y^2<\tfrac{1}{2}\},
\]
where~$\mathbb{D}$ is the unit disc. The \emph{Arctic circle} phenomenon~\cite{JPS-95,CEP-96} consists in the fact that the regions $\frac{1}{n}(p,q)\in \diamondsuit\smallsetminus \overline{\mathbb{D}}_\diamondsuit$ are asymptotically frozen -- the orientation of all dimers in these regions is almost deterministic -- while the liquid region~$\mathbb{D}_\diamondsuit$ carries long-range correlated fluctuations. A convenient way to encode these fluctuations is to consider random Thurston's height functions~\cite{thurston-height} of dimer covers. For homogeneous Aztec diamonds it is known~\cite{bufetov-gorin,CJY} that these height functions converge (in law) to the \emph{flat} Gaussian Free Field in~$\mathbb D$ provided that the following \emph{change of the radial variable} $\mathbb D_\diamondsuit\leftrightarrow \mathbb D$ is made:
\begin{equation}
\label{eq:rho-def}
r\,=\,{\sqrt{2}\rho}/(1\!+\!\rho^2),\quad r=(x^2\!+\!y^2)^{1/2}\in \bigl[0,\tfrac{\sqrt{2}}{2}\bigr]\ \ \leftrightarrow\ \ \rho\in [0,1].
\end{equation}

In recent developments~\cite{CLR2,CLR1}, the following idea arose: the conformal structure of dimer model fluctuations on abstract planar graphs can be understood via the so-called \emph{t-embeddings} (which appeared under the name {of} \emph{Coulomb gauges} in the preceding work~\cite{KLRR}; see also~\cite{affolter}) and the associated \emph{origami maps}; see Section~\ref{sec:discrete} for definitions. {The origami map of a t-embedding is a map from the plane to itself, so its graph is a {two-dimensional piecewise linear surface in a four-dimensional space; following~\cite{CLR2} we view the latter as the \emph{Minkowski} space $\R^{2,2}$ rather than as the Euclidean one. In the setup of this paper,} the image of the origami map becomes one-dimensional in the small mesh size limit, hence the limiting surface lives in the three-dimensional Minkowski space $\R^{2,1}\subset\R^{2,2}$.
It is shown in~\cite{CLR2} that, if the graphs of the origami maps converge to a \emph{space-like maximal surface in $\R^{2,1}$} as the mesh size tends to $0$, and provided that the technical assumptions~\textsc{Exp-Fat} and~\textsc{Lip} (see~\cite[Section~1.2.3]{CLR2}) about the non-degeneracy of the t-embeddings and origami maps are satisfied, then the intrinsic conformal metric of the limiting surface provides the right parametrization of the liquid region, i.e. the parametrization in which the fluctuations are given by a flat Gaussian Free Field.}

Recall that the three-dimensional Minkowski space $\R^{2,1}$ is equipped with a scalar product of signature $(2,1)$, which induces a positive-definite Riemannian metric on space-like surfaces. Space-like surfaces in~$\R^{2,1}$ with vanishing mean curvature are commonly known in the differential geometry literature under the name of \emph{maximal surfaces in the Minkowski space~$\R^{2,1}$} (e.g., see~\cite{kobayashi}), as they locally maximize the area.
An equivalent characterization of space-like maximal surfaces in~$\R^{2,1}$ which we will be using is the harmonicity of coordinate functions under a conformal parametrization of the surface by a complex parameter; this follows from the Weierstrass--Enneper representation of maximal surfaces \cite{kobayashi}. We stress that this conformal parameter is in general \emph{not} given by the limit of the t-embedding. The relevance of maximal surfaces with respect to the initial probabilistic model is that finding such a conformal parametrization will yield the right conformal structure to describe the GFF fluctuations. Indeed, in the limit, the dimer coupling functions become holomorphic in this conformal variable.

However, \cite{CLR2} does not provide any concrete setting where a space-like maximal surface is obtained as a limit of graphs of origami maps. The aim of this note is to provide a first example of such a setting, using the classical Aztec diamonds example. Namely, we give an inductive construction of t-embeddings and origami maps of those, and relate them to a \emph{discrete wave equation} in~2D. (It is worth noting that appearance of solutions of wave equations in the context of the dimer model on Aztec diamonds was observed in the foundational paper~\cite{CEP-96} already; see the discussion at the end of~\cite[Section~6.1]{CEP-96}.) Both analytic arguments and numerical simulations strongly support the convergence of origami maps to an explicit space-like maximal surface~$\mathrm{S}_\diamondsuit$ (see Fig.~\ref{fig:Aztec-Lorentz}), thus confirming the relevance of the setup developed in~\cite{CLR2,CLR1}. It is also worth mentioning that in~\cite[Section~4.2]{CLR2} the authors argue that such surfaces should appear in all setups in which the Gaussian structure of height fluctuations is expected.

\begin{remark} \label{rem:missing} In this short note we do \emph{not} rigorously prove the convergence of fundamental solutions of discrete wave equations to the continuous ones (see Eq.~\eqref{eq:f_0=o(1)} and~\eqref{eq:f_E-to-psi_E}). We also do not check the technical assumptions~\textsc{Exp-Fat} and~\textsc{Lip} for the t-embeddings and origamis of the Aztec diamonds. All these statements are discussed in a very recent paper~\cite{BNR}, which thus completes a new proof of the convergence of (gradients of) fluctuations in Aztec diamonds to the \emph{GFF in the conformal metric of~$S_\diamondsuit$}, following the general framework of~\cite{CLR2,CLR1}.
\end{remark}

\begin{remark} \label{rem:generalizations}
The recursive construction of the Aztec diamond via urban renewals that we use below was first introduced by Propp~\cite{propp-urban}. It provides an easy computation of the dimer partition function for the Aztec diamond with all weights equal to $1$. A stochastic variant of that construction, called domino shuffling \cite{EKLP-II}, samples a random dimer configuration of $A_{n+1}$ from a random dimer configuration of $A_n$ and additional iid Bernoulli random variables. In our analysis, the key role is played by the so-called Miquel (or central) move, which was discovered in~\cite{affolter,KLRR} to be the counterpart of the urban renewal move under the correspondence between \mbox{t-embeddings} and dimer models considered on abstract planar graphs. One may hope that other settings where generalizations of domino shuffling or urban renewal apply (e.g., see~\cite{BBBCCV,BMPW,speyer}) provide other explicit recursive constructions of t-embeddings. It would be interesting to further test the relevance of the framework of~\cite{CLR2,CLR1} on such examples, especially in presence of gaseous bubbles which, conjecturally, should lead to cusps of the limiting surface; see~\cite[Section~4.2]{CLR2}.
\end{remark}

\section{T-embeddings and origami maps of Aztec diamonds}
\label{sec:discrete}

In this section we construct the t-embeddings and origami maps of the Aztec diamonds. It is actually more convenient to work with the probabilistically equivalent \emph{reduced} Aztec diamonds.

\subsection{The Aztec diamond and its reduction}\label{sub:Aztec-comb} Let~$n\ge 1$ be a positive integer and consider the square grid $(\Z+\tfrac{1}{2})^2$ with vertices having half-integer coordinates. We define~$A_{n+1}$, the \emph{Aztec diamond} of size~$n+1$ to be a subgraph of this square grid formed by vertices~$(p,q)\in(\Z+\tfrac{1}{2})^2$ such that~$|p|+|q|\le n+1$. Independently of the parity of~$n$, we assume that the north-east ($p+q=n+1$) and south-west ($p+q=-n-1$) boundaries of~$A_{n+1}$ are composed of black vertices (while the north-west and south-east boundaries are composed of white ones); see~Fig.~\ref{fig:A-reduction}. We consider the homogeneous dimer model on~$A_{n+1}$, each edge has weight~$1$.

It is well known that correlation functions of the dimer model remain invariant under the following transforms (e.g., see~\cite[Fig.~1]{KLRR} or Fig.~\ref{fig:A-update} below):
\begin{itemize}
\item contraction of a vertex of degree~$2$ (if the two edge weights are equal to each other);
\item merging of parallel edges, the new edge weight is equal to the sum of initial ones;
\item gauge transformations, which amount to a simultaneous multiplication of weights of all edges adjacent to a given vertex by the same factor;
\item the so-called urban renewal move {(see Fig.~\ref{fig:urbanrenewal})}.
\end{itemize}

\begin{figure}
\includegraphics[width=0.7\textwidth]{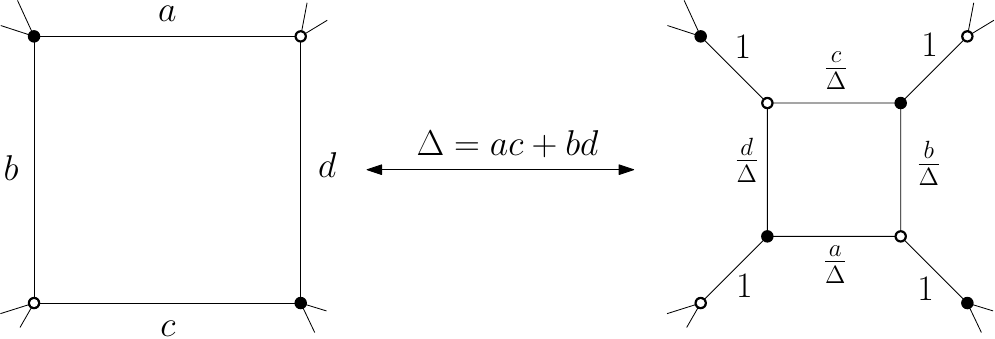}
\caption{{Local change of {a planar graph carrying the bipartite dimer model and of} the edge weights under the urban renewal move.}}\label{fig:urbanrenewal}
\end{figure}

\begin{figure}
\includegraphics[clip, trim=3.2cm 15.4cm 13.2cm 0.6cm, width=0.8\textwidth]{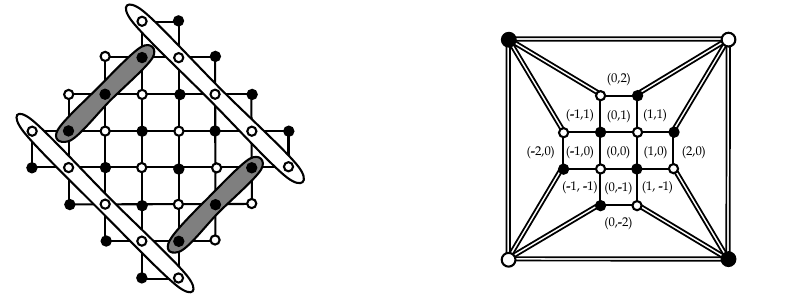}
\caption{Aztec diamond~$A_{n+1}$ and its reduction~$A'_{n+1}$ for~$n=3$. Double edges of~$A'_{n+1}$ have weight~$2$. The labeling~$(j,k)$ of faces of~$A'_{n+1}$ is shown, $|j|+|k|\!<\!n$.\label{fig:A-reduction}}
\end{figure}
Let~$A'_{n+1}$, the \emph{reduced Aztec diamond} of size~$n\!+\!1$, be obtained from~$A_{n+1}$ by the following sequence of moves (see Fig.~\ref{fig:A-reduction}):
\begin{enumerate}
\item contract the vertices~$(p,q)$ of~$A_{n+1}$ with~$p+q=n+1$, call~$w_{NE}$ the new (white) vertex;

\noindent contract the vertices~$(p,q)$ of~$A_{n+1}$ with~$p+q=-(n+1)$ to a new vertex~$w_{SW}$;

\item similarly, contract the vertices~$(p,q)$ of~$A_{n+1}$ with~$q-p=\pm (n+1)$ to new vertices~$b_{NW}$,~$b_{SE}$;

\item merge pairwise all the $4n$ pairs of parallel edges obtained during the first two steps.

\end{enumerate}
By construction, the reduced Aztec diamond~$A_{n+1}'$ is composed of~$A_{n-1}$ and four additional `boundary' vertices~$w_{NE}$, $b_{NW}$, $w_{SW}$ and~$b_{SE}$, which can be thought of as located at positions~$(\pm n,\pm n)$. The vertex~$w_{NE}$ is connected to all vertices on the north-east boundary of~$A_{n-1}$ as well as to~$b_{NW}$ and~$b_{SE}$ (and similarly for other boundary vertices). All edges of~$A_{n+1}'$ adjacent to boundary vertices have weight~$2$ while all other edges keep weight~$1$.

Clearly, the \emph{faces} (except the outer one) of the reduced Aztec diamond~$A'_{n+1}$ of size~$n+1$ can be naturally indexed by pairs~$(j,k)\in\mathbb Z^2$ such that~$|j|+|k|<n$; see Fig.~\ref{fig:A-reduction}.

We now describe how one can recursively construct reduced Aztec diamonds using the local transforms listed above. To initialize the procedure, note that the reduced Aztec diamond~$A'_2$ is a square formed by four outer vertices, with all edges having weight~$2$. To pass from~$A'_{n+1}$ to~$A'_{n+2}$, one uses the following operations (see Fig.~\ref{fig:A-update}, top):
\begin{enumerate}
\item split all edges adjacent to boundary vertices into pairs of parallel edges of weight~$1$;
\item perform the urban renewal moves with those faces of~$A'_{n+1}$ for which~$j+k+n$ is odd as shown on Fig.~\ref{fig:A-update}; note that the new vertical and horizontal edges have weight~$\frac12$;
\item contract all vertices of degree~$2$; note that now all edges adjacent to boundary vertices have weight~$1$ while all other edges have weight~$\frac12$;
\item multiply all edge weights by~$2$ (this is a trivial gauge transformation).
\end{enumerate}
\begin{figure}
\includegraphics[clip, trim=2.2cm 8.5cm 7.2cm 1.4cm, width=\textwidth]{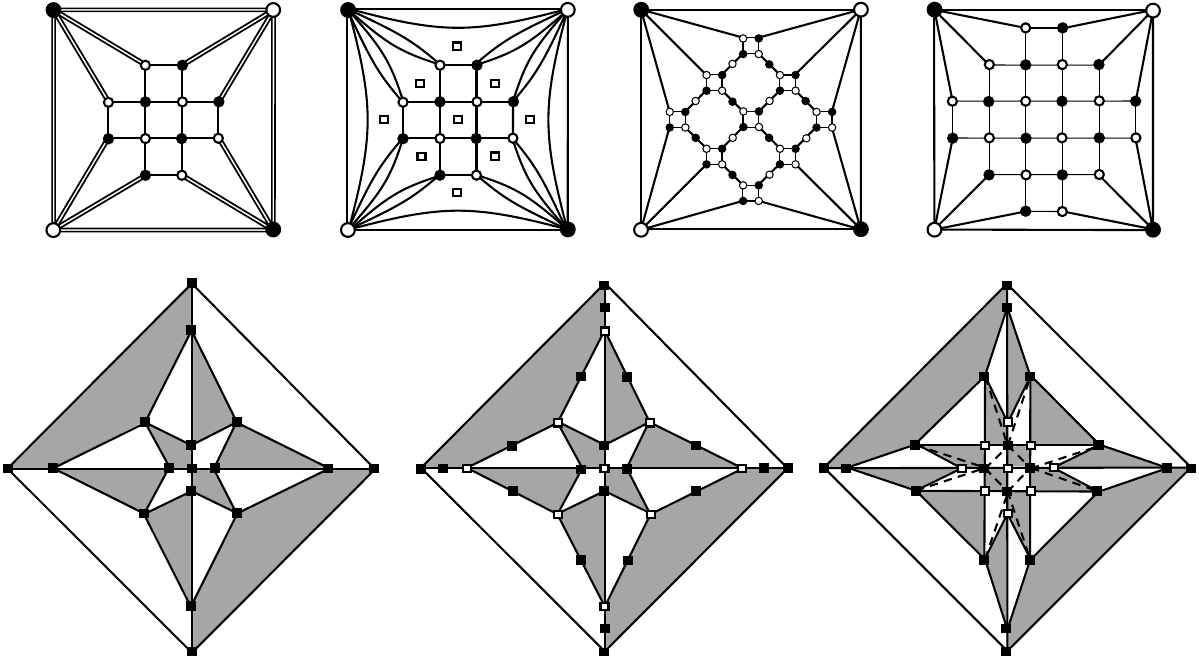}
\caption{A sequence of updates leading from~$A'_4$ to~$A'_5$. Top row: graphs carrying the dimer model; double, thick and thin edges have weight~$2$, $1$ and~$\frac12$, respectively. Bottom row: the same moves lead from the t-embedding~$\cT_3$ to~$\cT_4$; the last move (contraction of vertices of degree~$2$) removes dashed edges. Positions of urban renewals/central moves are indicated by white squares.\label{fig:A-update}}
\end{figure}

\subsection{Symmetric t-embeddings of reduced Aztec diamonds.}\label{sub:Aztec-temb} We now recall the construction introduced in~\cite{KLRR} under the name \emph{Coulomb gauges} and discussed in~\cite{CLR2,CLR1} under the name \emph{t-embeddings}; we refer to these papers for more details.
\begin{definition} Let~$\cG$ be a weighted planar bipartite graph and $\nu_{bw}$ denote the edge weights, where~$b$ and~$w$ stand for adjacent black and white vertices of~$\cG$, respectively. Let~$f$ be an inner face of~$G$ of degree~$2d$ with vertices denoted by~$w_1,b_1,\ldots,w_d,b_d$ in counterclockwise order; we set~$w_{d+1}:=w_1$. \emph{The $X$ variable} associated with~$f$ is defined as
\[
X_f\ :=\ \prod_{k=1}^d\frac{\nu_{b_kw_k}}{\nu_{b_kw_{k+1}}}\,.
\]
\end{definition}
Given a finite weighted planar bipartite graph~$\cG$ with a marked `outer' face~$f_\mathrm{out}$, let~$\cG^*$ denote the \emph{augmented dual} graph to~$\cG$, constructed as follows (see Fig.~\ref{fig:A-update}):
\begin{itemize}
\item to each inner face of~$\cG$, a vertex of~$\cG^*$ is associated;
\item $\deg f_\mathrm{out}$ vertices are associated to the outer face~$f_\mathrm{out}$ of~$\cG$;
\item to each edge of~$\cG$, a dual edge of~$\cG^*$ is associated;
\item an additional cycle of length~$\deg f_\mathrm{out}$ connecting the outer vertices is included into~$\cG^*$ (note that this cycle is not included into the graph~$\cG^*$ in the notation of~\cite{CLR1}).
\end{itemize}

\begin{definition} A \emph{t-embedding} of a finite weighted planar bipartite graph~$\cG$ is an embedding $\cT:\cG^*\to\C$ of its augmented dual $\cG^*$ such that the following conditions are satisfied:
\begin{itemize}
 \item under the embedding~$\cT$, the edges of~$\cG^*$ are non-degenerate straight segments, the faces are convex and do not overlap; the outer face of~$\cT(\cG^*)$ corresponds to the cycle replacing~$f_\mathrm{out}$ in the augmented dual~$\cG^*$;
 \item angle condition: for each inner vertex~$v^*$ of~$\cT(\cG^*)$, the sum of the angles at the corners corresponding to black faces is equal to $\pi$ (and similarly for white faces);
 \item moreover, for each inner face of~$\cG$, if~$v^*$ denotes the corresponding dual vertex in $\cG^*$ with neighbors $v^*_1,\ldots,v^*_{2d}$ such that the (dual) edge~$v^*v^*_{2k-1}$ is dual to the edge~$w_kb_k$ of~$\cG$ and~$v^*v^*_{2k}$ is dual to the edge~$b_kw_{k+1}$ (as above, we assume that~$w_1,b_1,\ldots,w_k,b_k$ are listed in counterclockwise order and~$w_{d+1}:=w_1$), then we have
\begin{equation}
\label{eq:Xf=dT/dT}
X_f\ =\ (-1)^{d+1}\prod_{k=1}^d\frac{\cT(v^*)-\cT(v^*_{2k-1})}{\cT(v^*_{2k})-\cT(v^*)}\,.
\end{equation}
\end{itemize}
\end{definition}
It is worth noting that the equation~\eqref{eq:Xf=dT/dT} implies that the \emph{geometric} weights (dual edge lengths)~$|\cT(v^*)-\cT(v^*_j)|$ are gauge equivalent to~$\nu_{bw}$. Therefore, in order to study the bipartite dimer model on an \emph{abstract} planar graph~$\cG$ one can first find a t-embedding of this graph and then consider the bipartite dimer model on this embedding with geometric weights. In our paper we apply this general philosophy to (reduced) Aztec diamonds~$A'_{n+1}$.

\begin{remark} In~\cite{CLR2}, the notion of a \emph{perfect} t-embedding of finite graphs is introduced. The additional condition required at boundary vertices of~$\cT(\cG^*)$ is that the outer face is a tangential polygon and that, for all outer vertices, the inner edges of~$\cT(\cG^*)$ are bisectors of the corresponding angles (in other words, the lines containing these edges pass through the center of the inscribed circle). In particular, the symmetric t-embeddings of reduced Aztec diamonds (defined in the next paragraph) are perfect.
\end{remark}

It was shown in~\cite{KLRR} that, if $\cG$ is a planar graph with outer face of degree $4$ and if we prescribe the four outer vertices of $\cG^*$ to be mapped to the four vertices of a given convex quadrilateral, then there exist two (potentially equal) t-embeddings of $G$ with these prescribed boundary conditions.  To respect the symmetries of the reduced Aztec diamonds~$A'_{n+1}$, we will study their t-embeddings~$\cT_n$ with symmetric boundary conditions. Namely, if~$v^*_E$, $v^*_N$, $v^*_W$ and~$v^*_S$ denote the four outer vertices of the augmented dual to~$A'_{n+1}$, then we require
\begin{equation}
\label{eq:temb-bc}
\cT_n(v^*_E)=1,\qquad \cT_n(v^*_N)=i,\qquad \cT_n(v^*_W)=-1,\qquad \cT_n(v^*_S)=-i;
\end{equation}
see Fig.~\ref{fig:Aztec123} for t-embeddings~$\cT_1$, $\cT_2$, $\cT_3$ of reduced Aztec diamonds~$A'_2$, $A'_3$, $A'_4$ satisfying~\eqref{eq:temb-bc}.
As explained in~\cite{KLRR}, the notion of t-embeddings is fully compatible with the local moves in the dimer model. Therefore, we can inductively construct t-embeddings of the reduced Aztec diamond~$A'_{n+1}$ with boundary conditions~\eqref{eq:temb-bc} starting with that of~$A'_2$. It is easy to see that the boundary conditions~\eqref{eq:temb-bc} yield~$\cT_1(0,0)=0$. In particular, this implies that~\eqref{eq:temb-bc} defines a t-embedding of~$A'_{n+1}$ uniquely.

Recall that the faces of~$A'_{n+1}$ are labeled by pairs~$(j,k)\in\mathbb{Z}^2$ such that~$|j|+|k|<n$ and that we denote by~$\cT_n$ the t-embedding of the reduced Aztec diamond~$A'_{n+1}$.
\begin{proposition}\label{prop:Tn-recursion} For each~$n\ge 1$, the t-embeddings~$\cT_{n+1}$ and~$\cT_n$ of the reduced Aztec diamonds~$A'_{n+2}$ and~$A'_{n+1}$ are related as follows. The positions~$\cT_{n+1}(j,k)$ are given by
\begin{enumerate}
\item if~$\{|j|,|k|\}=\{0,n\}$, then
\[
\begin{array}{rlrl}
\cT_{n+1}(-n,0)\!\!\!\!\!&=\frac{1}{2}\big(\cT_n(-n\!+\!1,0)+\cT_n(v_W^*)\big),&
\cT_{n+1}(0,n)\!\!\!\!\!&=\frac{1}{2}\big(\cT_n(0,n\!-\!1)+\cT_n(v_N^*)\big),\\[5pt]
\cT_{n+1}(0,-n)\!\!\!\!\!&=\frac{1}{2}\big(\cT_n(0,-n\!+\!1)+\cT_n(v_S^*)\big),&
\cT_{n+1}(n,0)\!\!\!\!\!&=\frac{1}{2}\big(\cT_n(n\!-\!1,0)+\cT_n(v_E^*)\big);
\end{array}
\]
\item if $1\le j\le n-1$ and~$k=n-j$, then
\[
\cT_{n+1}(j,n\!-\!j)=\tfrac{1}{2}\big(\cT_n(j\!-\!1,n\!-\!j)+\cT_n(j,n\!-\!j\!-\!1)\big)
\]
and similarly for~$\cT_{n+1}(-j,j-n)$, $\cT_{n+1}(-j,n-j)$ and~$\cT_{n+1}(j,j-n)$;
\smallskip
\item if~$|j|+|k|<n$ and~$j+k+n$ is even, then~$\cT_{n+1}(j,k)=\cT_n(j,k)$;
\smallskip
\item if~$|j|+|k|<n$ and~$j+k+n$ is odd, then
\[
\cT_{n+1}(j,k)=\tfrac{1}{2}\big(\cT_{n+1}(j\!+\!1,k)+\cT_{n+1}(j\!-\!1,k)+\cT_{n+1}(j,k\!+\!1)+\cT_{n+1}(j,k\!-\!1)\big)-\cT_n(j,k).
\]
\end{enumerate}
\end{proposition}
\begin{proof} When constructing the reduced Aztec diamond~$A'_{n+2}$ from~$A'_{n+1}$ following the procedure described in Section~\ref{sub:Aztec-comb}, the first step is to split edges of~$A'_{n+1}$ that are adjacent to boundary vertices into two parallel edges of equal weight. This amounts to adding the midpoints of the corresponding edges of~$\cT_n$ to the t-embedding, i.e., to~(1) and~(2); see the first step in Fig.~\ref{fig:A-update}.

Further, (3) reflects the fact that, for~$j+k+n$ even, the inner faces~$(j,k)$ of~$\cT_n$ are not destroyed by urban renewal moves and hence the positions of the corresponding dual vertices in the t-embedding~$\cT_{n+1}$ remain the same as in~$\cT_n$; see Fig.~\ref{fig:A-update}.

Finally, (4) follows from the explicit description of the urban renewal procedure in terms of the t-embedding, which is called the central (or Miquel) move in~\cite[Section~5]{KLRR}. In general, the position of~$\cT_{n+1}(j,k)$ is given by a \emph{rational} function involving the positions, already known, of four neighbors~$\cT_{n+1}(j\pm 1,k)$, $\cT_{n+1}(j,k\pm 1)$ and the position~$\cT_n(j,k)$. However, in our case we can additionally use the fact that $X_f=1$ for all $X$ variables assigned to faces at which the urban renewals are performed. Therefore, equation~\eqref{eq:Xf=dT/dT} implies that~$\cT_{n+1}(j,k)$ and~$\cT_n(j,k)$ are the two roots of the quadratic equation
\[
(z-\cT_{n+1}(j+1,k))(z-\cT_{n+1}(j-1,k))+(z-\cT_{n+1}(j,k+1))(z-\cT_{n+1}(j,k-1))=0.
\]
The required \emph{linear} expression for~$\cT_{n+1}(j,k)$ easily follows from Vieta's formula.
\end{proof}

\begin{figure}
\includegraphics[clip,trim=2.8cm 6.4cm 3.6cm 1.4cm, width=\textwidth]{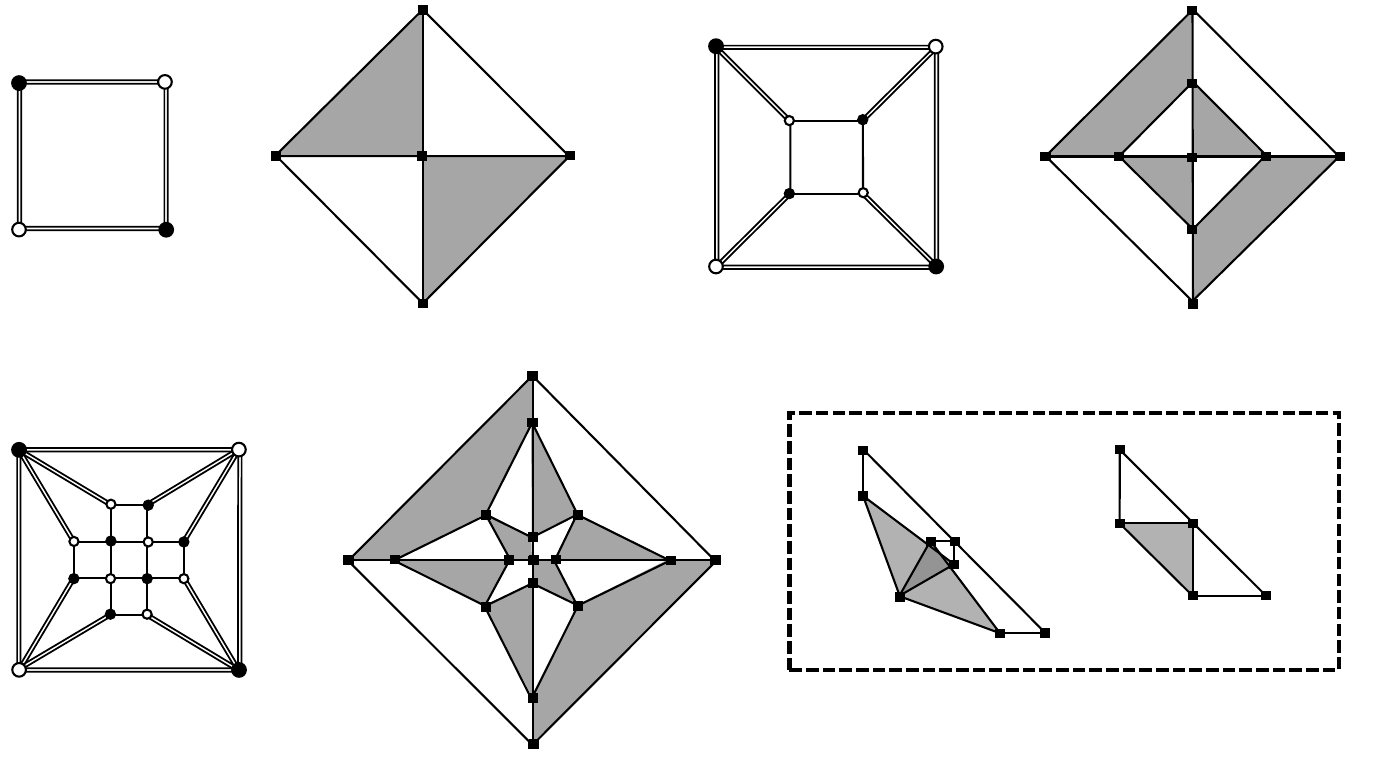}
\caption{Reduced Aztec diamonds~$A'_2$, $A'_3$, $A'_4$ and their symmetric \mbox{t-embeddings} $\cT_1$, $\cT_2$, $\cT_3$. \textsc{Bottom-right:} illustration of the origami maps $\cO_3$ and~$\cO_2$; the colors of faces inherit those from the north-east parts of~$\cT_3$ and~$\cT_2$. As discussed below (see Remark~\ref{rem:O'def} and Eq.~\eqref{eq:f_0=o(1)}), the images of the origami maps~$\cO_n$ are asymptotically one-dimensional for large~$n$.\label{fig:Aztec123}}
\end{figure}

\subsection{Origami maps associated to the t-embeddings of reduced Aztec diamonds}\label{sub:origami}
To each t-embedding~$\cT:\cG^*\to\C$ one can associate the so-called \emph{origami map}~$\cO:\cG^*\to\C$, defined up to an additive constant and a global unimodular pre-factor. Informally speaking, this map can be constructed as follows: choose a reference white face~$w_0$ of~$\cT(\cG^*)$ and set~$\cO(v^*):=\cT(v^*)$ for all vertices~$v^*$ of~$\cG^*$ adjacent to~$w_0$. To find positions~$\cO(v^*)$ of vertices of~$\cG^*$ lying at distance~$1$ from~$w_0$, glue copies of \emph{reflected} black faces adjacent to~$w_0$ to the image~$\cO(w_0)$. Continue this procedure inductively by gluing copies of white (resp., black) faces of $\cT(\cG^*)$ to the already constructed part of~$\cO(\cG^*)$, reverting the orientation of black ones and keeping the same adjacency relations between faces as in $\cT(\cG^*)$. The angle condition guarantees that this procedure is locally (and hence globally) consistent; see Fig.~\ref{fig:Aztec123} for an example.

We refer the reader to~\cite{KLRR} and~\cite[Section~2]{CLR1} for the formal definition of the origami map, note that this definition is also fully consistent with the local moves in the dimer model. In particular, an urban renewal move does not affect positions $\cO(v^*)$ except that of the dual vertex associated with the face of~$\cG$ at which this move is performed. Moreover, it is easy to see that the identity~\eqref{eq:Xf=dT/dT} and the angle condition imply that
\begin{equation}\label{eq:Xf=dO/dO}
X_f\ =\ (-1)^{d+1}\prod_{k=1}^d\frac{\cO(v^*)-\cO(v^*_{2k-1})}{\cO(v^*_{2k})-\cO(v^*)}\,,
\end{equation}
where we use the same notation as in~\eqref{eq:Xf=dT/dT}. Therefore, the update rule for the origami map~$\cO$ under an urban renewal is exactly the same as the central move for~$\cT$.

Let~$\cO_n$ be the origami map associated to the t-embedding~$\cT_n$ of the reduced Aztec diamond~$A'_{n+1}$ constructed starting with the north-east outer face~$w_{NE}$ of~$\cT_n$. In other words, we start the iterative construction of~$\cO_n$ by declaring~$\cO_n(v^*_E):=\cT_n(v^*_E)=1$, $\cO_n(v^*_N):=\cT_n(v^*_N)=i$ and~$\cO_n(j,n-1-j):=\cT_n(j,n-j-1)$ for all~$0\le j\le n-1$. It is easy to see that
\begin{equation}
\label{eq:origami-bc}
\cO_n(v^*_W)=\cO_n(v^*_E)=1\quad \text{and}\quad \cO_n(v^*_N)=\cO_n(v^*_S)=i
\end{equation}
for all~$n\ge 1$ and that~$\cO_1(0,0)=0$.
\begin{proposition}\label{prop:On-recursion} Let~$n\ge 1$. The origami map~$\cO_{n+1}$ can be constructed from~$\cO_n$ using the same update rules (1)--(4) as in Proposition~\ref{prop:Tn-recursion}, with boundary conditions~\eqref{eq:origami-bc} instead of~\eqref{eq:temb-bc}.
\end{proposition}
\begin{proof} The proof repeats that of Proposition~\ref{prop:Tn-recursion}. The least trivial update rule (4) for~$\cO_{n+1}$ follows from~\eqref{eq:Xf=dO/dO} in the same way as it follows from~\eqref{eq:Xf=dT/dT} for~$\cT_{n+1}$.
\end{proof}

\subsection{$\cT_n$ and~$\cO_n$ as solutions to the discrete wave equation} We now interpret the recurrence relations from Proposition~\ref{prop:Tn-recursion} as a discrete wave equation in the cone~$|j|+|k|<n$.

Denote by $\Z_+$ the set of all non-negative integers and by $\Lambda^+$ the cubic-centered half-space, that is the collection of all triples of integers $(j,k,n)\in\Z^2\times\Z_+$ such that $j+k+n$ is \emph{odd}. Let a function~$f_0:\Lambda^+\to\R$ be defined by the following conditions:
\begin{enumerate}
\item $f_0(j,k,0)=0$ for all~$j$ and~$k$ (such that~$j+k$ is odd);

\noindent $f_0(0,0,1)=1$ and $f_0(j,k,1)=0$ for all other pairs~$j,k$ (such that~$j+k$ is even);

\smallskip

\item for all~$(j,k,n+1)\in\Lambda^+$ such that~$n\ge 1$ we have
\begin{align}\label{eq:discretewave}
f(j,k,n+1)&+f(j,k,n-1)\\ &=\ \tfrac{1}{2}\big(f(j+1,k,n)+f(j-1,k,n)+f(j,k+1,n)+f(j,k-1,n)\big). \notag
\end{align}
\end{enumerate}

We call~\eqref{eq:discretewave} the \emph{discrete wave equation} and~$f_0$ is the \emph{fundamental solution} to this equation. To justify the name note that the formal subtraction of a (non-defined as~$j+k+n$ has the wrong parity) term~$2f(j,k,n)$ from both sides of~\eqref{eq:discretewave} allows one to see it as a discretization of the two-dimensional wave equation~$\partial_{tt}\psi\,=\,\tfrac{1}{2}(\partial_{xx}\psi+\partial_{yy}\psi)$.

\begin{definition}
Let $b_0$, $b_E$, $b_N$, $b_W$ and $b_S$ be five complex numbers. We say that a function $f:\Lambda^+\rightarrow\C$ solves the \emph{discrete wave equation in the cone $|j|+|k|<n$} with boundary conditions $(b_0;b_E,b_N,b_W,b_S)$ if it satisfies the following properties:
\begin{enumerate}\setcounter{enumi}{-1}
 \item $f(j,k,n)=0$ if~$|j|+|k|\ge n$ (in particular,~$f(j,k,0)=0$ for all~$j,k$);

 \smallskip

 \item boundary condition at the tip of the cone: $f(0,0,1)=b_0$;\

 \smallskip

 \item boundary conditions at the edges: for each~$n\ge 1$ we have
 \[
 \begin{array}{rlrl}
f(-n,0,n\!+\!1)\!\!\!&=\ \frac12\big(f(-n\!+\!1,0,n)+b_W\big),\ & f(0,n,n\!+\!1)\!\!\!&=\ \frac12\big(f(0,n\!-\!1,n)+b_N\big),\\[5pt]
f(0,-n,n\!+\!1)\!\!\!&=\ \frac12\big(f(0,-n\!+\!1,n)+b_S\big), & f(n,0,n\!+\!1)\!\!\!&=\ \frac12\big(f(n\!-\!1,0,n)+b_E\big).
\end{array}
\]
(Thus,~$f(n,0,n\!+\!1)=2^{-n}b_0+(1-2^{-n})b_E$ and similarly at other edges of the cone.)

\smallskip

\item discrete wave equation in the bulk of the cone and on its sides: for each~$(j,k,n+1)\in\Lambda^+$ such that \mbox{$n\ge 1$,} $|j|+|k|\le n$ and~$\{|j|,|k|\}\neq\{0,n\}$, the identity~\eqref{eq:discretewave} is satisfied. In particular, we require that
\[
f(j,n-j,n+1)=\frac{1}{2}(f(j-1,n-j,n)+f(j,n-j-1,n))
\]
for all~$1\le j\le n-1$ and similarly on other sides of the cone.
\end{enumerate}
\end{definition}
It is easy to see that the above conditions define the function~$f:\Lambda^+\to\C$ uniquely. Moreover, the solution in the cone with boundary conditions~$(1;0,0,0,0)$ is given by the fundamental solution~$f_0$ in the half-space~$\Lambda^+$. Denote by $f_E$ the solution in the cone with boundary conditions $(0;1,0,0,0)$ and similarly for~$f_N,f_W$ and~$f_S$. By linearity, a general solution in the cone with boundary conditions~$(b_0;b_E,b_N,b_W,b_S)$ can be written as
\[
f\ =\ b_0f_0+b_Ef_E+b_Nf_N+b_Wf_W+b_Sf_S.
\]
The functions~$f_E,f_N,f_W$ and~$f_S$ represent the impact of a constant source moving in the east, north, west and south directions, respectively, and thus can be expressed via~$f_0$. In particular,
\begin{equation}
\label{eq:fE-discr}
\textstyle f_E(j,k,n)\ =\ \frac{1}{2}\sum_{s=1}^{n-1} f_0(j-s,k,n-s)
\end{equation}
and similarly for $f_N,f_W$ and~$f_S$.
\begin{proposition} \label{prop:tembdiscretewave}
For all~$n\ge 1$, the values~$\cT_n(j,k)$ for~$j+k+n$ odd are given by the solution to the discrete wave equation in the cone~$|j|+|k|<n$ with boundary conditions~$(0;1,i,-1,-i)$.

Similarly, the values~$\cO_n(j,k)$ for~$j+k+n$ odd are given by the solution to the discrete wave equation in the cone~$|j|+|k|<n$ with boundary conditions~$(0;1,i,1,i)$.

If~$|j|+|k|<n$ and~$j+k+n$ is even, then~$\cT_n(j,k)=\cT_{n-1}(j,k)$ and~$\cO_n(j,k)=\cO_{n-1}(j,k)$.
\end{proposition}
\begin{proof}
This is a straightforward reformulation of recurrence relations from Proposition~\ref{prop:Tn-recursion}.
\end{proof}
\begin{remark} \label{rem:O'def} Let
\[
\cO'_n:=e^{i\frac{\pi}{4}}\big(\cO_n-\tfrac{1}{2}(1+i)\big)
\]
be another version of the origami map (recall that~$\cO_n$ is defined up to a global additive and a unimodular multiplicative constant only). It is easy to see that~$\sqrt{2}\operatorname{Im}(\cO'_n(j,k))=f_0(j,k,n)$ and that the values~$\sqrt{2}\operatorname{Re}(\cO'_n(j,k))$ are given by the solution to the wave equation in the cone with boundary conditions~$(0;1,-1,1,-1)$, provided that~$j+k+n$ is odd.
\end{remark}

\section{The limit of t-embeddings and the Aztec conformal structure}
\label{sec:cont}
Recall that the equation~\eqref{eq:discretewave} can be viewed as a discretization of the wave equation
\begin{equation}
\label{eq:contwave}
\partial_{tt}\psi\,=\,\tfrac{1}{2}(\partial_{xx}\psi+\partial_{yy}\psi).
\end{equation}
The continuous counterpart~$2\psi_0$ of the fundamental solution~$f_0$ satisfies boundary conditions $\psi_0\big|_{t=0}=0$, $\partial_t\psi_0\big|_{t=0}=\delta_{(0,0)}$, the factor~$2$ comes from the area of the horizontal plaquette with vertices~$(\pm 1,0,0)$,~$(0,\pm1,0)$. Therefore,
\begin{equation}
\label{eq:cont-fundsol}
\psi_0(x,y,t)=\frac{1}{\pi}\begin{cases}(t^2-2x^2-2y^2)^{-\frac12}& \text{if}~ x^2+y^2<\tfrac{1}{2}t^2,\\ 0 & \text{otherwise.}\end{cases}
\end{equation}
The continuous counterpart of the function~\eqref{eq:fE-discr} is thus given by a \emph{self-similar} solution
\[
\psi_E(x,y,t)\ =\ \int_0^t\psi_0(x-s,y,t-s)ds\ =:\ \psi_E^{(1)}(x/t,y/t)
\]
to the wave equation~\eqref{eq:contwave} with a \emph{super-sonic} source moving with speed~$1$ in the east direction. In particular,
\begin{equation}
\label{eq:psiE1=}
\psi_E^{(1)}(x,y)\ =\ \begin{cases}1 & \text{if}\ \ (x,y)\in\diamondsuit\smallsetminus\overline{\mathbb{D}}_\diamondsuit\ \ \text{and}\ \ x>\tfrac{1}{2}, \\ 0 & \text{if}\ \ (x,y)\in\diamondsuit\smallsetminus\overline{\mathbb{D}}_\diamondsuit\ \ \text{and}\ \ x<\tfrac{1}{2},\end{cases}
\end{equation}
where~$\diamondsuit=\{(x,y)\in\R^2\cong\C:|x|+|y|<1\}$ and~$\mathbb{D}_\diamondsuit=\frac{\sqrt{2}}{2}\mathbb D\subset\diamondsuit$. Below we take for granted the following two facts (see Remark~\ref{rem:missing} and~\cite{BNR} for more details):
\begin{itemize}
\item The discrete fundamental solution~$f_0$ uniformly (in $j,k$) decays in time:
\begin{equation}
\label{eq:f_0=o(1)}
\max\nolimits_{j,k:(j,k,n)\in\Lambda^+}f_0(j,k,n)=o(1)\quad \text{as}\ \ n\to\infty;
\end{equation}
\item The function~$f_E$ (and similarly for~$f_N,f_W$ and~$f_S$) has the following asymptotics:
\begin{equation}
\label{eq:f_E-to-psi_E}
f_E(j,k,n)=\psi_E^{(1)}(j/n,k/n)+o(1)\quad \text{as}\ \ n\to\infty,
\end{equation}
uniformly over~$(j/n,k/n)$ on compact subsets of~$\diamondsuit$.
\end{itemize}

We now consider the mapping
\begin{align}
\label{eq:def-z-theta}
(x,y)\in \mathbb D_\diamondsuit\ \ &\mapsto\ \ (z(x,y),\vartheta(x,y))\in\diamondsuit\times\R;\\
&z(x,y)\ :=\ \psi^{(1)}_E(x,y)+i\psi^{(1)}_N(x,y)-\psi^{(1)}_W(x,y)-i\psi^{(1)}_S(x,y), \notag\\
&\vartheta(x,y)\ :=\ \tfrac{\sqrt{2}}{2}\big(\psi^{(1)}_E(x,y)-\psi^{(1)}_N(x,y)+\psi^{(1)}_W(x,y)-\psi^{(1)}_S(x,y)\big), \notag
\end{align}
let us emphasize that~$z(x,y)\neq x+iy$. It is worth noting that this mapping is also well-defined on the whole set~$\diamondsuit$. Due to~\eqref{eq:psiE1=}, it sends the four connected regions of~$\diamondsuit\smallsetminus\mathbb{D}_\diamondsuit$ (which correspond to four frozen regions of Aztec diamonds) to \emph{points}~$(\pm 1,{1}/{\sqrt{2}})$ and $(\pm i,{1}/{\sqrt{2}})$, respectively.

Recall that~$\cT_n$ denotes the t-embedding of the reduced Aztec diamond~$A'_{n+1}$ and that~$\cO'_n$ is a version of the origami map associated to~$\cT_n$ introduced in Remark~\ref{rem:O'def}. Proposition~\ref{prop:tembdiscretewave} and \eqref{eq:f_0=o(1)},~\eqref{eq:f_E-to-psi_E} give
\[
\cT_n(j,k)\ =\ z(j/n,k/n)+o(1),\qquad
\begin{array}{rcl}
\operatorname{Re}(\cO'_n(j,k))&=&\vartheta(j/n,k/n)+o(1),\\[5pt]
\operatorname{Im}(\cO'_n(j,k))&=&o(1)
\end{array}
\]
as~$n\to\infty$, uniformly over~$j/n$ and~$k/n$ on compact subsets of~$\diamondsuit$.

It is not hard to find an explicit expression for~$z(x,y)$ and~$\vartheta(x,y)$. To this end, introduce the polar coordinates~$(r,\phi)$ in the plane~$(x,y)$. A straightforward computation shows that each self-similar solution~$\psi^{(1)}(x/t,y/t)$ of the wave equation~\eqref{eq:contwave} satisfies
\begin{equation}
\label{eq:wave_polar}
\biggl[(1-2r^2)\partial_{rr}+\biggl(\frac{1}{r}-4r\biggr)\partial_r+\frac{1}{r^2}\partial_{\phi\phi}\biggr]\psi^{(1)}(x,y)\ =\ 0.
\end{equation}
This equation can be viewed as the harmonicity condition provided an appropriate change of the radial variable is made. Namely, let
\[
\zeta(x,y):=e^{i\phi(x,y)}\rho(x,y)\in\mathbb{D}\quad \text{for}\ \ (x,y)\in\mathbb{D}_\diamondsuit,
\]
where the new radial variable~$\rho=\rho(x,y)$ is defined by~\eqref{eq:rho-def}. (Note that in our context this change of variables naturally appears from the 2D wave equation and not from any pre-knowledge on the dimer model.) It is straightforward to check that~\eqref{eq:wave_polar} can be written as
\begin{align*}
\frac{\rho^2}{r^2}\Delta_{(\rho,\phi)}\psi^{(1)}\ =\ \frac{\rho^2}{r^2}\biggl[\frac{1}{\rho}\partial_\rho(\rho\partial_\rho)+\frac{1}{\rho^2}\partial_{\phi\phi}\biggr]\psi^{(1)}\ =\ 0.
\end{align*}
Thus, both functions~$z$ and~$\vartheta$ are harmonic in the variable~$\zeta$ and hence can be written in terms of their boundary values. Namely, we have
\begin{align*}
z(x,y)&=\operatorname{hm}_{\mathbb{D}}(\zeta(x,y);\gamma_E)+i\operatorname{hm}_{\mathbb{D}}(\zeta(x,y);\gamma_N)- \operatorname{hm}_{\mathbb{D}}(\zeta(x,y);\gamma_W)-i\operatorname{hm}_{\mathbb{D}}(\zeta(x,y);\gamma_S),\\[5pt]
\vartheta(x,y)&=\tfrac{\sqrt{2}}{2}\big(\operatorname{hm}_{\mathbb{D}}(\zeta(x,y);\gamma_E)-\operatorname{hm}_{\mathbb{D}}(\zeta(x,y);\gamma_N)+ \operatorname{hm}_{\mathbb{D}}(\zeta(x,y);\gamma_W)-\operatorname{hm}_{\mathbb{D}}(\zeta(x,y);\gamma_S)\big),
\end{align*}
where~$\gamma_E:=(e^{-i\frac{\pi}{4}},e^{i\frac{\pi}{4}})$, $\gamma_N:=(e^{i\frac{\pi}{4}},e^{i\frac{3\pi}{4}})$, $\gamma_W:=(e^{i\frac{3\pi}{4}},e^{i\frac{5\pi}{4}})$ and $\gamma_S:=(e^{i\frac{5\pi}{4}},e^{i\frac{7\pi}{4}})$ denote four boundary arcs of the unit disc~$\mathbb{D}$ {and $\operatorname{hm}_{\mathbb{D}}(\cdot;\gamma)$ denotes the harmonic measure of the boundary arc $\gamma$}. Recall that, if~$\alpha<\beta<\alpha+2\pi$, then
\begin{equation}
\label{eq:hm=}
\operatorname{hm}_\mathbb{D}(\zeta;(e^{i\alpha},e^{i\beta}))\ =\ \tfrac{1}{\pi}\big(\arg(e^{i\beta}-\zeta)-\arg(e^{i\alpha}-\zeta)\big)-\tfrac{1}{2\pi}(\beta-\alpha)\,.
\end{equation}

Note that~$(x,y)\mapsto \zeta(x,y)$ is an orientation preserving diffeomorphism from~$\mathbb{D}_\diamondsuit$ to~$\mathbb{D}$. From the explicit formula for~$z(x,y)$ and the argument principle, it is also clear that~$\zeta\mapsto z$ is an orientation preserving diffeomorphism from~$\mathbb{D}$ to~$\diamondsuit$. In particular, it makes sense to view~$\vartheta=\vartheta(x,y)$ as a function of~$z=z(x,y)\in\diamondsuit$. Since each of the origami maps~$\cO_n$ does not increase distances comparing to those in~$\cT_n$, the same should hold true in the limit as~$n\to\infty$. Therefore, the function~$z\mapsto \vartheta$ is~$1$-Lipschitz.

Following the approach developed in~\cite{CLR2,CLR1}, we now consider the surface~$(z(x,y),\vartheta(x,y))$ as embedded into the \emph{Minkowski} space~$\R^{2,1}$ rather than into the usual Euclidean space~$\R^3$. The Lipschitzness condition discussed above implies that this surface is space-like. We are now ready to formulate the most conceptual observation of this note. Let~$\mathrm{C}_\diamondsuit$ be a (non-planar) quadrilateral in~$\R^2\times \R\cong\C\times\R$ with vertices
\begin{equation}
\label{eq:C=}
\mathrm{C}_\diamondsuit:\ \ (1,\tfrac{\sqrt{2}}{2})\ \ \text{---}\ \ (i,-\tfrac{\sqrt{2}}{2})\ \ \text{---}\ \ (-1,\tfrac{\sqrt{2}}{2})\ \ \text{---}\ \ (-i,-\tfrac{\sqrt{2}}{2})\ \ \text{---}\ \ (1,\tfrac{\sqrt{2}}{2}).
\end{equation}
\begin{proposition} The surface~$\mathrm{S}_\diamondsuit:=\{(z(x,y),\vartheta(x,y)),\ (x,y)\in\mathbb{D}_\diamondsuit\}$ is a space-like maximal surface in the Minkowski space~$\R^{2,1}$, whose boundary is given by~$\mathrm{C}_\diamondsuit$; see~Fig.~\ref{fig:Aztec-Lorentz}. Moreover, the variable~$\zeta(x,y)\in\mathbb{D}$ provides the conformal parametrization of this maximal surface.
\end{proposition}
\begin{proof} From now onwards we view~$z$ and~$\vartheta$ as functions of the complex variable~$\zeta=\zeta(x,y)\in\mathbb{D}$ rather than as functions of~$(x,y)\in\mathbb{D}_\diamondsuit$. Let us first discuss the boundary behavior of the mapping~$\zeta\mapsto (z,\vartheta)$. From the explicit formulas, it is clear that~$(z(\zeta),\vartheta(\zeta))\to (1,\frac{\sqrt{2}}{2})$ if~$\zeta$ approaches the boundary arc~$\gamma_E$, and similarly for the three other arcs. Near the point~$\zeta=e^{i\frac{\pi}{4}}$ separating~$\gamma_E$ and~$\gamma_N$ one observes the following behavior:
\[
\begin{array}{rcl}
z(\zeta)&=&\frac{1+i}{2}+\frac{1-i}{\pi}\big(\arg(e^{i\frac{\pi}{4}}-\zeta)-\tfrac{\pi}{4}\big) +O(|\zeta-e^{i\frac{\pi}{4}}|),\\[5pt]
\vartheta(\zeta)&=&\tfrac{\sqrt{2}}{\pi}\big(\arg(e^{i\frac{\pi}{4}}-\zeta)-\tfrac{\pi}{4}\big)+O(|\zeta-e^{i\frac{\pi}{4}}|),
\end{array}
\]
where the branch of~$\arg$ is chosen so that~$\arg(e^{i\frac{\pi}{4}}-\zeta)\in (-\frac{\pi}{4},\frac{3\pi}{4})$ for~$\zeta\in\mathbb D$. This means that~$\vartheta$ and~$z$ are almost linearly dependent near the point~$e^{i\frac{\pi}{4}}$ and thus the boundary of the surface~$\mathrm{S}_\diamondsuit$ contains the full segment with endpoints~$(1,\tfrac{\sqrt{2}}{2})$ and~$(i,-\tfrac{\sqrt{2}}{2})$. Repeating the same argument for each of the four points separating the arcs~$\gamma_E$, $\gamma_N$, $\gamma_W$ and~$\gamma_S$, we conclude that the boundary of~$\mathrm{S}_\diamondsuit$ is given by~\eqref{eq:C=}.

\begin{figure}
\begin{minipage}[b]{0.45\textwidth}
\includegraphics[clip, trim=4.5cm 7.6cm 1cm 4cm, width=\textwidth]{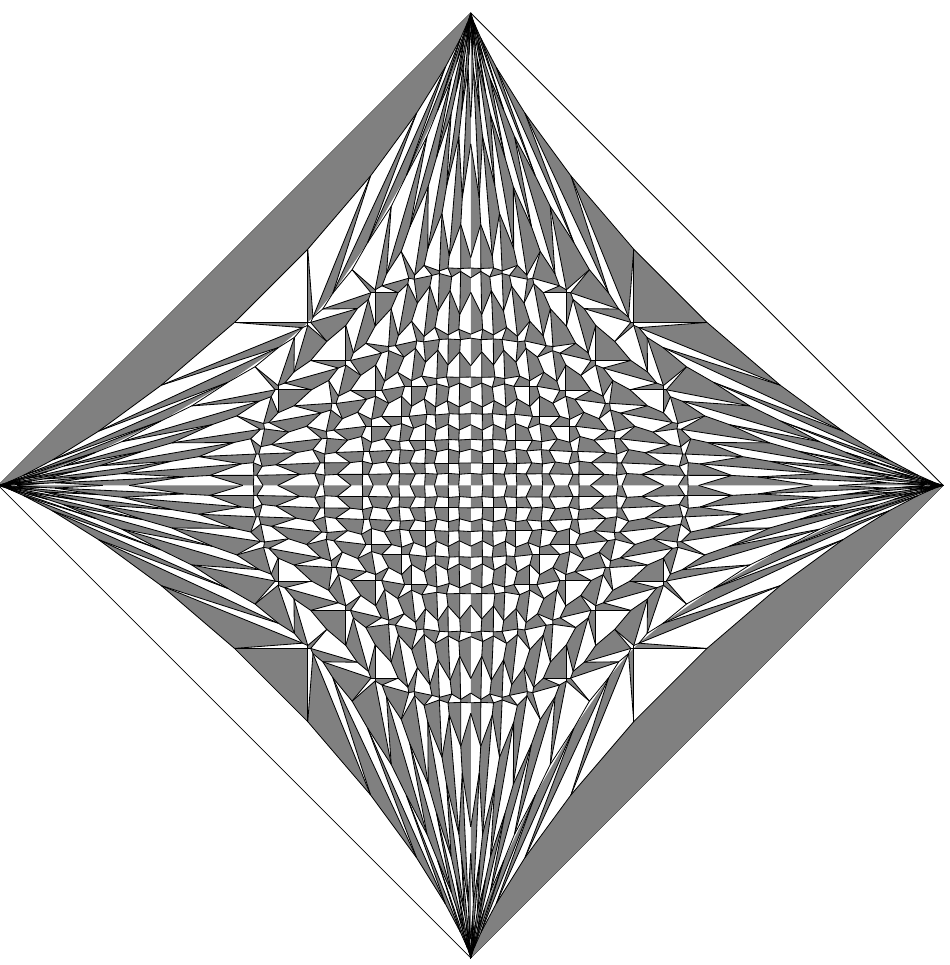}
\end{minipage}
\hskip 0.09\textwidth
\begin{minipage}[b]{0.45\textwidth}
\includegraphics[clip, trim=4.5cm 7.6cm 1cm 4cm, width=\textwidth]{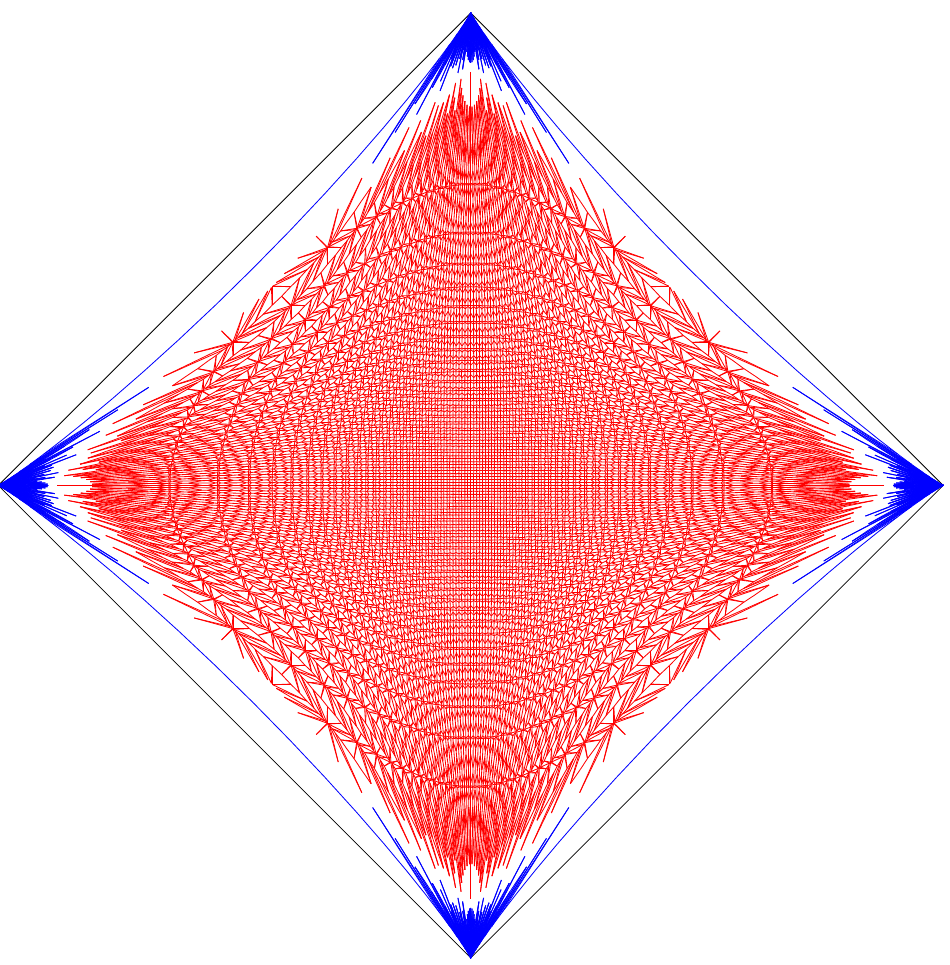}
\end{minipage}
\caption{\textsc{Left:} (symmetric) t-embedding of the Aztec diamond of size~$27$. \textsc{Right:}~\mbox{t-embedding} of the Aztec diamond of size~$102$, the edges connecting~$\cT_{101}(j_1,k_1)$ and~$\cT_{101}(j_2,k_2)$ are colored \emph{red} if both~$(j_{1,2}^2+k_{1,2}^2)\le 0.49\cdot 10^4$, \emph{blue} if both~$(j_{1,2}^2+k_{1,2}^2)\ge 0.5\cdot 10^4$, and are not shown otherwise.\label{fig:A25bw-A100}}
\end{figure}
\begin{figure}

\begin{minipage}[c]{0.32\textwidth}
\includegraphics[clip, trim=4.5cm 7.6cm 1.2cm 4cm, width=\textwidth]{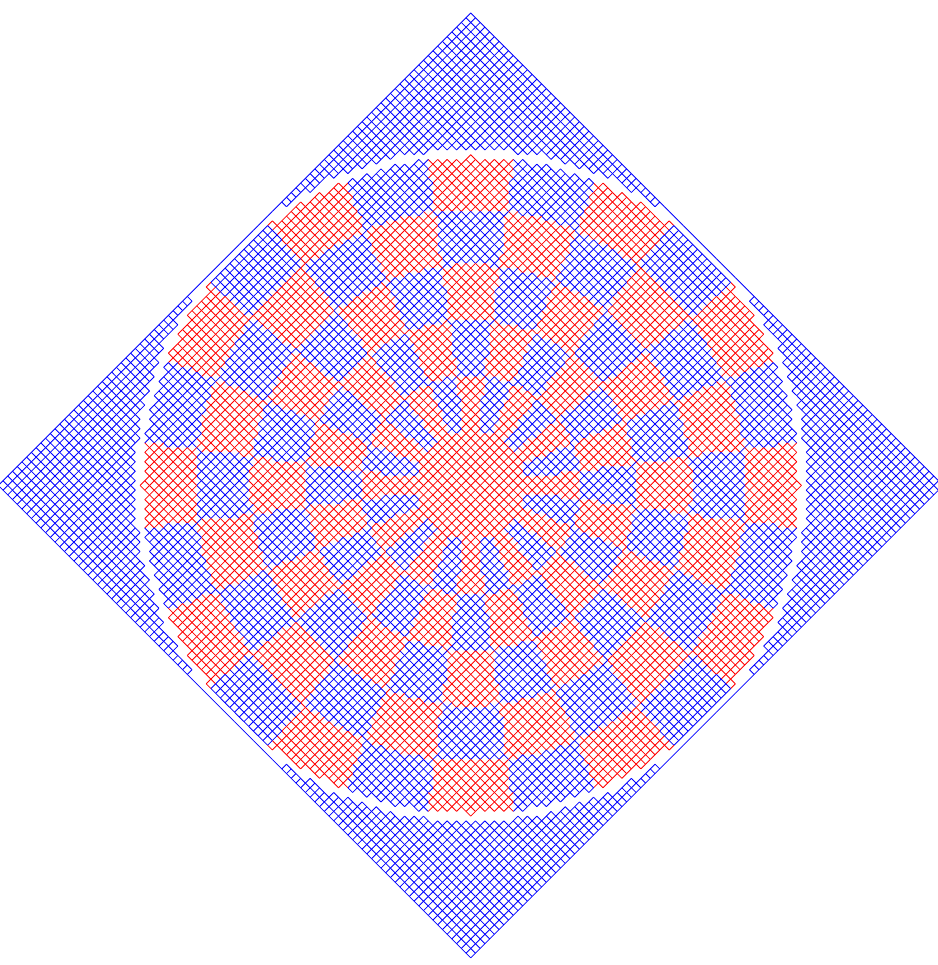}

\includegraphics[clip, trim=4.5cm 7.6cm 1.2cm 4cm, width=\textwidth]{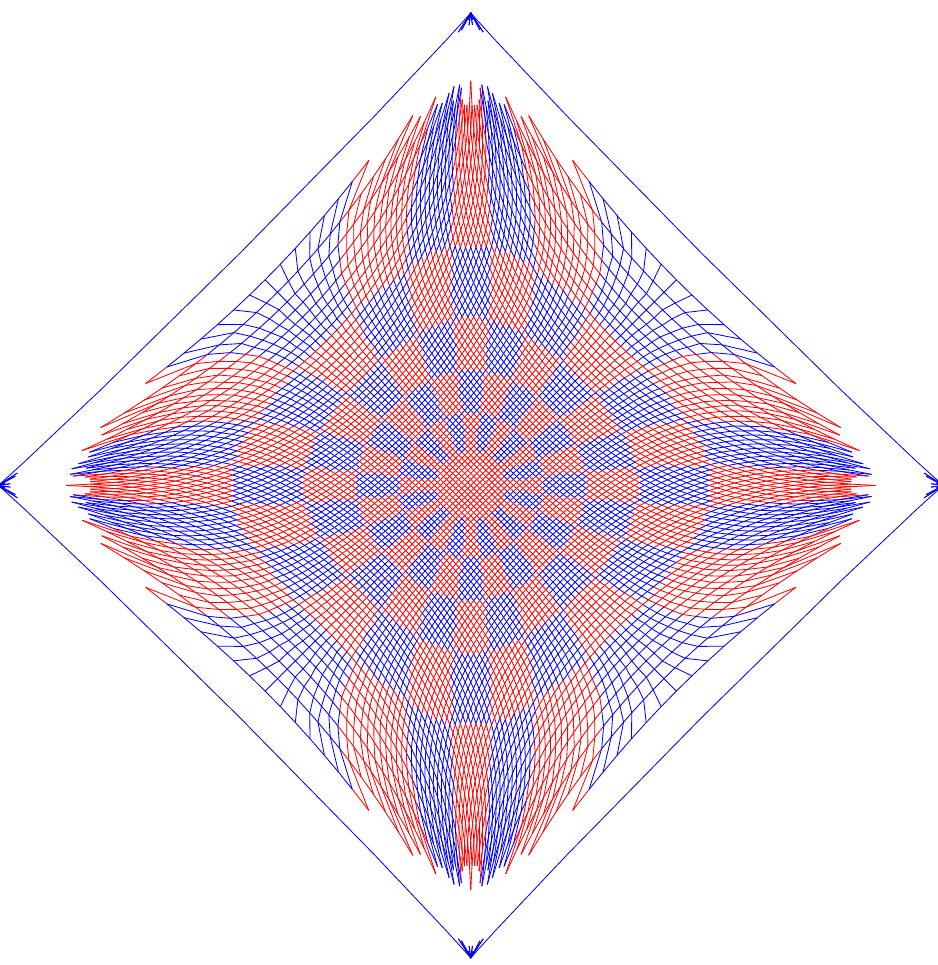}
\end{minipage}
\hskip 0.07\textwidth
\begin{minipage}[c]{0.56\textwidth}
\includegraphics[clip, trim=6.4cm 16cm 5.6cm 4.7cm, width=\textwidth]{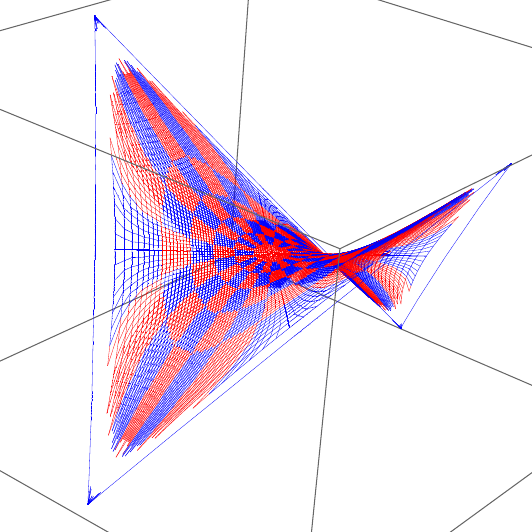}
\end{minipage}
\caption{\textsc{Top-left:} the grid of points~$(\frac{1}{1600}j,\frac{1}{1600}k)$,~$j,k\in16\mathbb{Z}$, \mbox{$j+k\in 32\mathbb{Z}$}. The additional red-blue coloring of the liquid region is introduced for \mbox{visibility}. \textsc{Bottom-left:} The image of this grid under the mapping~$\cT_{1601}$. Four frozen zones of~$A_{1602}$ are collapsed to tiny vicinities of~the~corners; the image of a thin white vicinity of the arctic circle is stretched in a square-root fashion.
\textsc{Right:}~The grid of points~$(\cT_{1601}(j,k),\cO'_{1601}(j,k))$,~$j,k\in16\mathbb{Z}$, \mbox{$j+k\in 32\mathbb{Z}$}, \mbox{approximates} the space-like maximal surface~$\mathrm{S}_\diamondsuit$ spanning the contour~$\mathrm{C}_\diamondsuit\subset{\R^{2,1}}$.\label{fig:Aztec-Lorentz}}
\end{figure}

We now prove that~$\zeta$ is a \emph{conformal} parametrization of~$\mathrm{S}_\diamondsuit$; recall that the angles on the surface~$\mathrm{S}_\diamondsuit$ are measured with respect to the Minkowski metric in~$\R^{2,1}$. To this end, we need to show that
\begin{equation}\label{eq:conf}
(\partial_{\zeta}z)(\zeta)\cdot(\partial_{\zeta} \overline{z})(\zeta)\ =\ ((\partial_{\zeta}\vartheta)(\zeta))^2\quad \text{for}\ \ \zeta\in\mathbb{D},
\end{equation}
where~$\partial_{\zeta}=\frac{1}{2}(\partial_{\mathrm{Re}\zeta}-i\partial_{\mathrm{Im}\zeta})$ stands for the Wirtinger derivative with respect to the variable~$\zeta$. Recall that the functions~$z(\zeta)$ and~$\vartheta(\zeta)$ are explicit linear combinations of the harmonic measures of the arcs~$\gamma_E$, $\gamma_N$, $\gamma_W$, $\gamma_S$. From~\eqref{eq:hm=}, one easily sees that
\[
\partial_{\zeta}\operatorname{hm}_\mathbb{D}(\zeta;(e^{i\alpha},e^{i\beta}))\ =\ \tfrac{i}{2\pi}\big((\zeta-e^{i\beta})^{-1}-(\zeta-e^{i\alpha})^{-1}\big)\,.
\]
Therefore, the equation~\eqref{eq:conf} is equivalent to the identity
\begin{align*}
\big((1-i)(\zeta-e^{i\frac{\pi}{4}})^{-1}+(1+i)(\zeta-e^{i\frac{3\pi}{4}})^{-1}+ (-1+i)(\zeta-e^{i\frac{5\pi}{4}})^{-1}+(-1-i)(\zeta-e^{i\frac{7\pi}{4}})^{-1}\big)&\\
\times \big((1+i)(\zeta-e^{i\frac{\pi}{4}})^{-1}+(1-i)(\zeta-e^{i\frac{3\pi}{4}})^{-1}+ (-1-i)(\zeta-e^{i\frac{5\pi}{4}})^{-1}+(-1+i)(\zeta-e^{i\frac{7\pi}{4}})^{-1}\big)&\\
 =\ 2\big((\zeta-e^{i\frac{\pi}{4}})^{-1}-(\zeta-e^{i\frac{3\pi}{4}})^{-1}+ (\zeta-e^{i\frac{5\pi}{4}})^{-1}-(\zeta-e^{i\frac{7\pi}{4}})^{-1}\big)^2&\,,
\end{align*}
which is straightforward to check.

Finally, both~$z$ and~$\vartheta$ are \emph{harmonic} functions in the conformal parametrization~$\zeta$ of the space-like surface~$\mathrm{S}_\diamondsuit$ (provided that the metric on~$\mathrm{S}_\diamondsuit$ is induced from~$\R^{2,1}$). Classically, this implies that~$\mathrm{S}_\diamondsuit$ is maximal in the metric of~$\R^{2,1}$; see \cite{kobayashi}.
\end{proof}

\section{Numerical simulations} The discrete wave equation~\eqref{eq:discretewave} trivially admits very fast simulations, a few examples are given in Fig.~\ref{fig:A25bw-A100} and Fig.~\ref{fig:Aztec-Lorentz}. Actually, constructing solutions to~\eqref{eq:discretewave} is more memory-consuming than time-consuming if one wants to keep all binary digits so as not to lose control over cancellations, inherent to the wave equation in~2D. The pictures in Fig.~\ref{fig:A25bw-A100} and, notably, in Fig.~\ref{fig:Aztec-Lorentz} are obtained by such \emph{exact} simulations.

Despite some clear resemblance, we stress that the space-like maximal surface is not the same as the limiting height function, also known as the limit shape.

\section{Conclusion} In this note we studied the symmetric t-embeddings~$\cT_n$ of homogeneous Aztec diamonds~$A_n$ and tested the framework developed in~\cite{CLR2,CLR1} on this classical example of the dimer model that leads to a non-trivial conformal structure of the fluctuations. Both analytic arguments and numerical simulations strongly indicate the convergence of the graphs~$(\cT_n,\cO_n)$ of the corresponding origami maps~$\cO_n$ to a space-like maximal surface~$\mathrm{S}_\diamondsuit$ embedded into the Minkowski space~$\R^{2,1}$. The intrinsic conformal structure of~$\mathrm{S}_\diamondsuit$ provides a new description of the well-studied scaling limit of dimer fluctuations in the liquid regions of~$A_n$. Though additional work is required to get a new rigorous proof of the convergence of fluctuations (see Remark~\ref{rem:missing}), our results strongly support the paradigm of~\cite{CLR2,CLR1}. See also the very recent paper~\cite{BNR}, which provides the missing ingredients and thus completes such a proof.

\subsection*{Acknowledgements} This research was partially supported by the ANR-18-CE40-0033 project DIMERS. We would like to thank Beno\^{\i}t Laslier and Marianna Russkikh for many discussions and comments on draft versions of this note, and Jesper Lykke Jacobsen, Peter McGrath, Gregg Musiker and Istv\'an Prause for pointing out relevant references. D.C. is grateful to Olivier Biquard for {a very helpful advice on the Lorentz geometry.} S.R. thanks the Fondation Sciences Math\'ematiques de Paris for the support during the academic year 2018/19.

\Addresses

\begin{thebibliography}{10}

\bibitem{affolter}
Niklas~C. {Affolter}.
\newblock {Miquel Dynamics, Clifford Lattices and the Dimer Model}.
\newblock {{\em Lett. Math. Phys.}, 111(3):Paper No. 61,23, 2021.}
\newblock {\url{https://doi.org/10.1007/s11005-021-01406-0}}

\bibitem{ADSV}
Nicolas Allegra, J\'{e}r\^{o}me Dubail, Jean-Marie St\'{e}phan, and Jacopo
  Viti.
\newblock Inhomogeneous field theory inside the arctic circle.
\newblock {\em J. Stat. Mech. Theory Exp.}, (5):053108, 76, 2016.
\newblock {\url{https://doi.org/10.1088/1742-5468/2016/05/053108}}

\bibitem{BBBCCV}
{Dan} Betea, {C\'edric} Boutillier, {J\'er\'emie} Bouttier, {Guillaume} Chapuy, {Sylvie} Corteel, and
  {Mirjana} Vuleti\'{c}.
\newblock Perfect sampling algorithms for {S}chur processes.
\newblock {\em Markov Process. Related Fields}, 24(3):381--418, 2018.
\newblock {\url{https://arxiv.org/pdf/1407.3764.pdf}}

\bibitem{BMPW}
Mireille Bousquet-M\'{e}lou, James Propp, and Julian West.
\newblock Perfect matchings for the three-term {G}ale-{R}obinson sequences.
\newblock {\em Electron. J. Combin.}, 16(1):Research Paper 125, 37, 2009.
\newblock {\url{http://www.combinatorics.org/Volume_16/Abstracts/v16i1r125.html}}

\bibitem{bufetov-gorin}
Alexey Bufetov and Vadim Gorin.
\newblock Fluctuations of particle systems determined by {S}chur generating
  functions.
\newblock {\em Adv. Math.}, 338:702--781, 2018.
\newblock {\url{https://doi.org/10.1016/j.aim.2018.07.009}}

\bibitem{BNR}
Tomas Berggren, Matthew Nicoletti and Marianna Russkikh.
\newblock Perfect t-embeddings of uniformly weighted Aztec diamonds and tower graphs.
\newblock {\em Int. Math. Res. Not. IMRN}, (7):5963--6007, 2024.
\newblock {\url{https://doi.org/10.1093/imrn/rnad299}

\bibitem{CLR2}
Dmitry Chelkak, Beno\^{\i}t Laslier, and Marianna Russkikh.
\newblock {Bipartite dimer model: perfect t-embeddings and Lorentz-minimal surfaces.}
\newblock {\em arXiv e-prints}, arXiv:2109.06272, Sep 2021.}
\newblock {\url{https://arxiv.org/pdf/2109.06272.pdf}}

\bibitem{CLR1}
Dmitry {Chelkak}, Beno{\^\i}t {Laslier}, and Marianna {Russkikh}.
\newblock {Dimer model and holomorphic functions on t-embeddings of planar
  graphs}.
\newblock {\em Proc. Lond. Math. Soc. (3)}, 126(5):1656--1739, 2023.
\newblock {\url{https://doi.org/10.1112/plms.12516}}


\bibitem{CJY}
Sunil Chhita, Kurt Johansson, and Benjamin Young.
\newblock Asymptotic domino statistics in the {A}ztec diamond.
\newblock {\em Ann. Appl. Probab.}, 25(3):1232--1278, 2015.
\newblock {\url{https://doi.org/10.1214/14-AAP1021}}


\bibitem{CEP-96}
Henry Cohn, Noam Elkies, and James Propp.
\newblock Local statistics for random domino tilings of the {A}ztec diamond.
\newblock {\em Duke Math. J.}, 85(1):117--166, 1996.
\newblock {\url{https://doi.org/10.1215/S0012-7094-96-08506-3}}

\bibitem{EKLP-I}
Noam Elkies, Greg Kuperberg, Michael Larsen, and James Propp.
\newblock Alternating-sign matrices and domino tilings. {I}.
\newblock {\em J. Algebraic Combin.}, 1(2):111--132, 1992.
\newblock {\url{https://doi.org/10.1023/A:1022420103267}}

\bibitem{EKLP-II}
Noam Elkies, Greg Kuperberg, Michael Larsen, and James Propp.
\newblock Alternating-sign matrices and domino tilings. {II}.
\newblock {\em J. Algebraic Combin.}, 1(3):219--234, 1992.
\newblock {\url{https://doi.org/10.1023/A:1022483817303}}

\bibitem{gorin-lectures}
Vadim Gorin.
\newblock {{\em Lectures on random lozenge tilings}, volume 193.
\newblock Cambridge: Cambridge University Press, 2021.}
\newblock {\url{https://doi.org/10.1017/9781108921183}}

\bibitem{GBDJ}
Etienne Granet, Louise Budzynski, J\'{e}r\^{o}me Dubail, and Jesper~Lykke
  Jacobsen.
\newblock Inhomogeneous {G}aussian free field inside the interacting arctic
  curve.
\newblock {\em J. Stat. Mech. Theory Exp.}, (1):013102, 31, 2019.
\newblock {\url{https://doi.org/10.1088/1742-5468/aaf71b}}

\bibitem{JPS-95}
William {Jockusch}, James {Propp}, and Peter {Shor}.
\newblock {Random Domino Tilings and the Arctic Circle Theorem}.
\newblock {\em arXiv Mathematics e-prints}, page math/9801068, Jan 1998.
\newblock {\url{https://arxiv.org/pdf/math/9801068.pdf}}

\bibitem{kenyon-lectures}
Richard Kenyon.
\newblock Lectures on dimers.
\newblock In {\em Statistical mechanics}, volume~16 of {\em IAS/Park City Math.
  Ser.}, pages 191--230. Amer. Math. Soc., Providence, RI, 2009.
  \newblock {\url{https://doi.org/10.1090/pcms/016/04}}

\bibitem{KLRR}
Richard {Kenyon}, Wai~Yeung {Lam}, Sanjay {Ramassamy}, and Marianna {Russkikh}.
\newblock {Dimers and Circle patterns}.
\newblock {{\em Ann. Sci. \'Ec. Norm. Sup\'er. (4)}, 55(3):863--901, 2022.}
\newblock {\url{https://doi.org/10.24033/asens.2507}}

\bibitem{kobayashi}
{Osamu Kobayashi.
\newblock Maximal surfaces in the {$3$}-dimensional {M}inkowski space {$L\sp{3}$}.
\newblock {\em Tokyo J. Math.}, 6(2):297--309, 1983.}
\newblock {\url{https://doi.org/10.3836/tjm/1270213872}}

\bibitem{propp-urban}
James Propp.
\newblock Generalized domino-shuffling.
\newblock {{\em Theoret. Comput. Sci.},} volume 303, pages 267--301. 2003.
\newblock Tilings of the plane.
\newblock {\url{https://doi.org/10.1016/S0304-3975(02)00815-0}}


\bibitem{speyer}
David~E. Speyer.
\newblock Perfect matchings and the octahedron recurrence.
\newblock {\em J. Algebraic Combin.}, 25(3):309--348, 2007.
\newblock {\url{https://doi.org/10.1007/s10801-006-0039-y}}

\bibitem{thurston-height}
William~P. Thurston.
\newblock Conway's tiling groups.
\newblock {\em Amer. Math. Monthly}, 97(8):757--773, 1990.
\newblock {\url{https://doi.org/10.2307/2324578}}

\end{thebibliography}
\end{document}